\documentclass{article}[12pt]

\usepackage{amssymb, amsthm, amsmath,graphicx}
\usepackage{a4wide}

\newcommand{\fra}[2]{\textstyle{\frac{#1}{#2}}}
\newcommand{\beqn}{\begin{eqnarray}\begin{aligned}}
\newcommand{\eqn}{\end{aligned}\end{eqnarray}}

\usepackage{setspace}
\usepackage{natbib}

\newtheorem{thm}{Theorem}[section]

\newtheorem{lem}[thm]{Lemma}

\usepackage{pstricks,pst-node,pst-tree,pst-text,pstricks-add}

\begin{document}

\newcommand\zerosplit{\Tdot  \Tdot }
\newcommand\nextstep{{\LARGE $\rightarrow$}}

\begin{titlepage}

\begin{center}

{\Large {\bf Markov invariants and the isotropy\\ subgroup of a quartet tree}}

\vspace{2em}

J G Sumner and P D Jarvis$^{\dagger}$  
\par \vskip 1em \noindent
{\it School of Mathematics and Physics, University of Tasmania, TAS 7001, Australia}\\
\end{center}
\par \vskip .3in \noindent

\vspace{1cm} \noindent\textbf{Abstract} \normalfont 
\\\noindent
The purpose of this article is to show how the isotropy subgroup of leaf permutations on binary trees can be used to systematically identify tree-informative invariants relevant to models of phylogenetic evolution.
In the quartet case, we give an explicit construction of the full set of representations and describe their properties.   
We apply these results directly to Markov invariants, thereby extending previous theoretical results by systematically identifying linear combinations that vanish for a given quartet.
We also note that the theory is fully generalizable to arbitrary trees and is equally applicable to the related case of phylogenetic invariants.
All results follow from elementary consideration of the representation theory of finite groups.

\vfill
\hrule \mbox{} \\
{\footnotesize 
{$^{\dagger}$ Alexander von Humboldt Fellow}\\
{\textit{keywords}: invariants, phylogenetics, group characters, branching rules}\\
{\textit{corresponding author}: Jeremy Sumner, jsumner@utas.edu.au}\\
UTAS-PHYS-2008-15
}
\end{titlepage}


%

\section{Preliminaries}
Phylogenetic methods seek to reconstruct the evolutionary history of organisms from present-day data such as DNA and are of fundamental importance in the biological sciences \citep{felsenstein2004}.
Approaches to this important problem draw upon sophisticated mathematical, statistical and computational techniques (see \citet{gascuel2005} for an overview).
From a purely theoretical point of view, this represents a wonderful confluence of hitherto disparate areas of mathematics. 
In particular, models of phylogenetic evolution require a marriage between graph theory, combinatorics and stochastic processes (a comprehensive treatment can be found in \citet{semple2003}). 
There is also a rich algebraic structure underlying phylogenetic models -- particularly when the complications of working with binary trees is taken into account.
For instance, spectral analysis of the Kimura 3ST model using Hadamard conjugation \citep{hendy1989} and group based approaches to phylogenetic invariants \citep{evans1993} provide novel applications of algebra to phylogenetics.
This article serves as a direct sequel to the algebraic approach applying group representation theory to phylogenetics given in \citet{sumner2008}, where ``Markov invariants'' were defined and explored.

Standard stochastic models of phylogenetic evolution are high-dimensional, with the number of free parameters being proportional to the number of leaves on the evolutionary tree.
Given that DNA sequences are of finite extent, it follows that phylogenetic data sets are often quite sparse and significant model-fitting problems arise with respect to the issue of bias/variance trade-off \citep{burnham2002}.
In this light, Markov invariants provide \emph{one-dimensional} ``representations'' of these stochastic models that retain some of the complex structure of these models, while greatly reducing the number of free parameters present.  
Significantly, Markov invariants are defined to respect the infinitesimal unfolding of a continuous-time Markov chain.
This property is not stipulated in the definition of phylogenetic invariants and there is some evidence (given in \citet{sumner2008}) that this additional structure can assist in the search for ``powerful'' sets of phylogenetic invariants \citep{eriksson2008}.
In particular, it should be noted that the popular Log-Det pairwise distance \citep{steel1994} has as its foundation the simplest example of a Markov invariant.

We say that a Markov invariant is ``tree-informative'' if it satisfies the conditions of a phylogenetic invariant \citep{cavender1987,lake1987} for particular trees.
Here we show how to systematically find linear combinations of Markov invariants that are tree-informative.
An explicit construction is given in the case of quartet trees by studying the irreducible representations of the isotropy subgroup of leaf permutations on quartets. 

Presently we review some basic concepts and terminology from \citet{sumner2008}.
 
Given a group $\mathcal{G}$, recall that a \emph{group representation} is a homomorphism $\rho:\mathcal{G} \rightarrow GL(V)$, where $GL(V)$ is the set of invertible linear operators on a vector space $V$.
This provides an \emph{action} of $\mathcal{G}$ on $V$ and in this case $V$ is referred to as a $\mathcal{G}$-\emph{module} (or, a module of $\mathcal{G}$, or, when the group is understood, simply, a module).
$U\subseteq V$ is said to form an \emph{invariant subspace} if it is closed under the action of $\mathcal{G}$, i.e. $\rho(\mathcal{G})\cdot U\subseteq U$. 

In this article, a tree $\mathcal{T}$ is a connected acyclic graph with vertices of valence 3 or 1 only.
The vertices of valence 1 are referred to as \emph{leaves} and are denoted by $L$ with $m\!:=\!|L|$. 
All results given will be relevant to the \emph{general Markov model} \citep{allman2003} of sequence evolution on a tree (including the IID assumptions), with the additional constraint that all transition matrices are chosen from the Markov semigroup \citep{sumner2008}.
Restricting to the Markov semigroup ensures that the process arises as a continuous-time Markov chain, and allows us to refer to notions of continuity and the infinitesimal.
We denote elements of the Markov semigroup as $M_a$ and employ \emph{right} multiplication so that the matrix element $m^{(a)}_{ji}\!:=\!\left[M_a\right]_{ji}$ represents the probability of a transition $i\rightarrow j$.

In particular, consider random variables defined at the leaves of a tree $X_1,X_2,\ldots, X_m$.
We suppose these random variables take on one of $k$ discrete values with an associated probability distribution 
\beqn
p_{i_1i_2\ldots i_m}:=\mathbb{P}\left[X_1\!=\!i_1,X_2\!=\!i_2,\ldots X_m\!=\!i_m\right].\nonumber
\eqn 
Given the $k$-dimensional vector space $V\cong \mathbb{C}^k$ with basis vectors $\{e_i\}_{1\leq i\leq k}$, the \emph{phylogenetic tensor} $P\in V^{\otimes m}$ is defined as
\beqn
P:=\sum_{1\leq i_1,i_2,\ldots ,i_m\leq k}p_{i_1i_2\ldots i_m}e_{i_1}\otimes e_{i_2}\otimes \ldots \otimes e_{i_m}.\nonumber
\eqn
If this distribution is generated under a Markov assumption (as is standard for phylogenetic models), the ``local'' (no branching events) change of this tensor is described by
\beqn
\label{eq:time-evol}
P'=g\cdot P:=M_1\otimes M_2\otimes \ldots \otimes M_m\cdot P,
\eqn
where each $M_i$ is an element of the Markov semigroup.
Markov invariants (of weight $w$) are defined as functions that take a simple form under this local change:
\beqn
f(P'):=f(g\cdot P)=\det(g)^w f(P).\nonumber
\eqn
As each term in $\det(g)\!=\!\det(M_1)\ldots \det(M_m)$ can be related to expected number of state changes under the model \citep[chap. 8]{semple2003}, we see that a Markov invariant reduces the high-dimensionality of (\ref{eq:time-evol}) to a single parameter that is related to the total number of state changes expected from this process.  
However, as it stands, this definition of Markov invariants says nothing about any underlying tree structure.
It is rectifying this situation that is the main purpose of this article. 

The definition can be viewed as a group action on the Markov invariants themselves by setting $\left(g^{-1}\circ f\right)(P)\!:=\!f(gP)$. 
Thus a Markov invariant transforms under the Markov process as a one-dimensional module of the Markov semigroup:
\beqn
g^{-1}\circ f=\det(g)^wf.\nonumber
\eqn
It should be noted that existence of $g^{-1}$ is guaranteed as all elements of the Markov semigroup are invertible as linear operators (we return to this point in the next section).

By applying Schur-Weyl duality between the symmetric and the general linear groups, existence conditions for such invariants were given in \citet{sumner2008} using inner multiplications of Schur functions.
In particular, in the case of DNA and quartet trees, $k\!=\!4$ and $m\!=\!4$, it was shown that there exist four linearly independent Markov invariants of degree $d\!=\!5$. 

In this article we extend these results by including the ``global'' aspect of the tree and branching process thereof.
Previously this has been achieved by laboriously checking (with a computer) for linear relations between Markov invariants when evaluated on canonical forms of phylogenetic tensors arising from different trees. 
This procedure identified linear combinations of Markov invariants that vanish for certain trees, hence producing tree-informative invariants that satisfy the usual definition of phylogenetic invariants along with respecting the local transformation properties discussed above.
Here we will achieve the same result by studying the transformation properties of Markov invariants under leaf permutations.

Rather than deal with the automorphism group of a tree \citep{godsil2001}, we consider the isotropy subgroup $\mathcal{G}_\mathcal{T}$ of leaf permutations $\mathfrak{S}_{m}\cong \text{Sym}(L)$.
Formally this corresponds to the automorphism group restricted to the leaf vertices:
\beqn
\mathcal{G}_\mathcal{T}\equiv \left. \text{Aut}(\mathcal{T})\right|_L .\nonumber
\eqn
Although it is clear that as abstract groups we have $\mathcal{G}_\mathcal{T}\cong \text{Aut}(\mathcal{T})$ (under the action of an element of $\text{Aut}(\mathcal{T})$ the images of the leaves uniquely determines the image of each internal vertex), it is crucial to our discussion to make this distinction so that $\mathcal{G}_\mathcal{T}$ can be viewed as a subgroup of the symmetric group $\mathfrak{S}_{m}$.
This allows us to define an action of $\mathcal{G}_\mathcal{T}$ on the space of phylogenetic tensors and respects the underlying biology, as it is the labelling of vertices at the leaves that is of primary importance.

In what follows we will deal with the simplest non-trivial case: quartets.
We will derive the multiplication table for the isotropy group of a quartet, compute its conjugacy classes, irreducible representations, character table, and group branching rule upon restriction from $\mathfrak{S}_4$.
In doing so we completely characterize the quartet case and give a clear path to the general theory for larger trees.
All results are applied to Markov invariants, but it should be noted that the technique presented is directly relevant to other structures that arise in phylogenetics including, of course, phylogenetic invariants.
\section{Isotropy subgroups of quartets}
\begin{figure}[tbp]
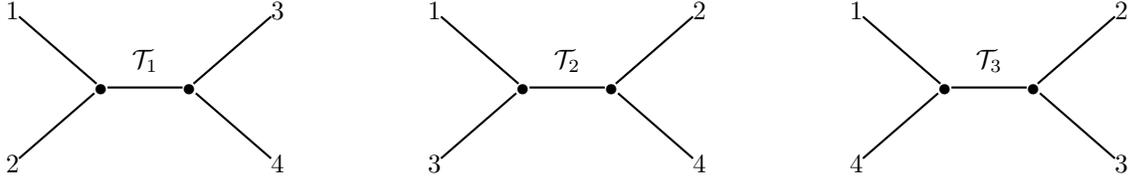

  \centering  
  $\psmatrix[colsep=1cm,rowsep=.6cm]
  1 & & & 3 \\
    & \bullet & \bullet   \\
  2 & & & 4  
  \ncline{2,2}{1,1}
  \ncline{2,2}{3,1}
  \ncline{2,3}{1,4}
  \ncline{2,3}{3,4}
  \ncline{2,2}{2,3}^{\mathcal{T}_1}
  \endpsmatrix
  $
  \hspace{1.7cm}
  $\psmatrix[colsep=1cm,rowsep=.6cm]
  1 & & & 2 \\
    & \bullet & \bullet   \\
  3 & & & 4  
  \ncline{2,2}{1,1}
  \ncline{2,2}{3,1}
  \ncline{2,3}{1,4}
  \ncline{2,3}{3,4}
  \ncline{2,2}{2,3}^{\mathcal{T}_2}
  \endpsmatrix
  $
  \hspace{1.7cm}
  $\psmatrix[colsep=1cm,rowsep=.6cm]
  1 & & & 2 \\
    & \bullet & \bullet   \\
  4 & & & 3  
  \ncline{2,2}{1,1}
  \ncline{2,2}{3,1}
  \ncline{2,3}{1,4}
  \ncline{2,3}{3,4}
  \ncline{2,2}{2,3}^{\mathcal{T}_3}
  \endpsmatrix
  $
    \caption{Unrooted, leaf-labelled quartet trees}
    \label{fig:quartets}
\end{figure}

Consider the three possible unrooted leaf-labelled quartet trees given in Figure~\ref{fig:quartets}.
We can represent each of these quartets as a word from the alphabet $\{$``1'',``2'',``3'',``4'',``$|$''$\}$ in several ways:
\beqn
\mathcal{T}_1&:=12|34\cong 21|34\cong 34|12\ldots ,\nonumber\\ 
\mathcal{T}_2&:=13|24\cong 31|24\cong 24|13\ldots,\\ 
\mathcal{T}_3&:=14|23\cong 41|23\cong 23|14\ldots.
\eqn 
An action of the symmetric group $\mathfrak{S}_4$ on these words is defined by permuting the leaf labels:
\beqn
ij|kl \mapsto\sigma \cdot ij|kl = \sigma(i)\sigma(j)|\sigma(k)\sigma(l),\quad \forall \sigma\in\mathfrak{S}_4.\nonumber
\eqn
For example, using the cycle notation for the symmetric group we have
\beqn
(12)\cdot \mathcal{T}_1&=(12)\cdot 12|34=21|34\cong 12|34=\mathcal{T}_1, \nonumber\\
(123)\cdot \mathcal{T}_1&=(123)\cdot 12|34=23|14\cong 14|23=\mathcal{T}_3,
\eqn
and
\beqn
(13)(24)\cdot\mathcal{T}_1=(13)(24)\cdot 12|34=34|12\cong 12|34=\mathcal{T}_1.\nonumber
\eqn
This group action actually defines a homomorphism $\mathfrak{S}_4 \rightarrow \mathfrak{S}_3 $, as $\mathfrak{S}_4$ acts by permuting the three quartets.
However, this homomorphism will not be of primary interest to us.

Given a group $\mathcal{G}$ acting on a set $X$, the \emph{isotropy subgroup} $\mathcal{G}_x$ of the element $x\in X$ is defined as the set of group elements that leave $x$ fixed:
\beqn
\mathcal{G}_x:=\left\{g\in \mathcal{G}\,|\,g\cdot x=x \right\}.\nonumber
\eqn
It is easy to show that $\mathcal{G}_x$ does indeed form a subgroup. 
(The reader should note that some authors refer to an isotropy subgroup as a ``stabilizer'' subgroup.)

We are interested in the isotropy subgroup of each of the quartet trees:
\beqn
\mathcal{G}_{12|34}:=\left\{\sigma\in \mathfrak{S}_4\,|\,\sigma \cdot 12|34\cong 12|34\right\},\nonumber
\eqn
with $\mathcal{G}_{13|24}$ and $\mathcal{G}_{14|23}$ defined similarly.
By exhaustive search through the elements of $\mathfrak{S}_4$, we find that
\beqn
\mathcal{G}_{12|34}=\left\{e,(12),(34),(12)(34),(13)(24),(14)(23),(1324),(1423)\right\},\nonumber
\eqn
where $e$ denotes the identity element.
This subgroup can be generated from the elements $(1324)$ and $(13)(24)$ so that any element can be expressed as a product of these two.
If we set $a\!=\!(1324)$ and $b\!=\!(13)(24)$ we find that $a^4\!=\!b^2\!=\!e$ and $b^{-1}ab\!=\!a^{-1}$.
In this way we see that $\mathcal{G}_{12|34}$ is isomorphic to the dihedral group $D_8$; the symmetry group of a square.
 
Recall that, for finite trees, a ``rotation'' is defined as an element of $\text{Aut}(\mathcal{T})$ (excluding the identity) that fixes at least one vertex of $\mathcal{T}$, whereas a ``reflection'' flips at least one internal edge \citep{gawron1999}.
Thus, referring to Figure~\ref{fig:quartets} we see that $(12)$, $(34)$ and $(12)(34)$ are rotations, while $(13)(24)$, $(14)(23)$, $(1324)$ and $(1423)$ are reflections.

In this article we consider phylogenetic tensors that are constructed using transition matrices chosen from the Markov semigroup.
Recall that every element $M$ of the Markov semigroup satisfies $0< \det(M)\leq 1$, with $\det(M)\!=\! 1$ occurring only in the trivial case where $M$ is the identity operator \citep{sumner2008}.
Thus if we assume that all transition matrices are non-trivial, thereby ensuring non-zero branch lengths and \emph{binary} evolutionary trees, we can apply identifiability of tree topology \citep{chang1996} and conclude that the phylogenetic tensors on quartets can be partitioned into disjoint subsets, with each subset corresponding to a quartet.
Thus, if we denote the set of phylogenetic tensors as $V^{\mathcal{T}_i}\subset V^{\otimes 4}$, where $\mathcal{T}_i$ is a quartet and $V\cong \mathbb{C}^k$, we have: 
\beqn
V^{\mathcal{T}_i}\cap V^{\mathcal{T}_j}=\emptyset,\quad \forall i\neq j. \nonumber
\eqn
It should be noted that these are sub\emph{sets} and clearly not sub\emph{spaces} of the vector space $V^{\otimes 4}$.
In fact, the recent non-identifiability result for phylogenetic mixtures of \citet{matsen2007} imply that each $V^{\mathcal{T}_i}$ is not even closed under real, convex linear combinations.
However, this will not affect any of the results discussed in the present work: we will simply have to replace the phrase ``invariant subspace'' with ``invariants subset'', where relevant.

There is an action of $\mathfrak{S}_4$ on $V^{\otimes 4}$ defined as
\beqn
\sigma\psi:=\sum_{i_1,\ldots ,i_4}\psi_{i_1i_2i_3i_4}e_{i_{\sigma(1)}}\otimes e_{i_{\sigma(2)}}\otimes  e_{i_{\sigma(3)}}\otimes  e_{i_{\sigma(4)}}\nonumber.
\eqn
Informally, this is equivalent to writing
\beqn\label{eq:perminformal}
\sigma\cdot \psi_{i_1i_2i_3i_4} =\psi_{i_{\bar{\sigma}(1)} i_{\bar{\sigma}(2)} i_{\bar{\sigma}(3)} i_{\bar{\sigma}(4)}}, 
\eqn
where, for ease of reading, we have set $\bar{\sigma}\equiv\sigma^{-1}$.
Clearly this induces an action of $\mathcal{G}_{12|34}$ on the set of phylogenetic tensors.

\begin{lem}
\label{lemma1}
$V^{\mathcal{T}_1}$ forms an invariant subset under the action of $\mathcal{G}_{12|34}$.
Further,
\beqn
\sigma V^{\mathcal{T}_2}&\subseteq V^{\mathcal{T}_2},\qquad
\sigma V^{\mathcal{T}_3}&\subseteq V^{\mathcal{T}_3},\nonumber
\eqn
if $\texttt{sgn}(\sigma)=1$, and 
\beqn
\sigma V^{\mathcal{T}_2}&\subseteq V^{\mathcal{T}_3},\qquad
\sigma V^{\mathcal{T}_3}&\subseteq V^{\mathcal{T}_2},\nonumber
\eqn
if $\texttt{sgn}(\sigma)=-1$, for all $\sigma\in \mathcal{G}_{12|34}$.
\end{lem}

\begin{proof}
This result follows easily by noting that $\mathcal{G}_{12|34}\cdot \mathcal{T}_1=\mathcal{T}_1$ by definition, and checking that $\sigma\cdot  \mathcal{T}_2=\mathcal{T}_2$ if $\texttt{sgn}(\sigma)=1$ and $\sigma\cdot  \mathcal{T}_2=\mathcal{T}_3$ if $\texttt{sgn}(\sigma)=-1$.
However, we confirm the proof explicitly to illustrate the way the symmetric group acts on phylogenetic tensors.

The components of any phylogenetic tensor $P\in V^{\mathcal{T}_1}$ can be expressed as 
\beqn
p_{i_1i_2i_3i_4}=\sum_{1\leq i,j\leq k}m^{(1)}_{i_1i}m^{(2)}_{i_2i}m^{(3)}_{i_3j}m^{(4)}_{i_4j}m^{(0)}_{ji}\pi_{i},\nonumber 
\eqn
where, for each $a$, $m^{(a)}_{ji}$ are the matrix elements of an element $M_a$ of the Markov semigroup. 
We have (arbitrarily) chosen to root the quartet at the parent vertex of leaf 1 and 2 with root distribution $\pi$ (see Figure~\ref{fig:quartettensor}).

The ``trimmed'' tensor $\widetilde{P}$ \citep{sumner2008} is generated from $P$ by trimming off the pendant edges of the tree or, more precisely, by setting each transition matrix on a pendant edge equal to the identity matrix:
\beqn
\label{eq:trimmed}
\widetilde{p}_{i_1i_2i_3i_4} =\sum_{1\leq i,j\leq k}\delta_{i_1i}\delta_{i_2i}\delta_{i_3j}\delta_{i_4j}m^{(0)}_{ji}\pi_{i}= \delta_{i_1i_2}\delta_{i_3i_4} m^{(0)}_{i_3i_1}\pi_{i_1}.
\eqn
We can write $P=M_1\otimes M_2 \otimes M_3 \otimes M_4\cdot \widetilde{P}$, and observe that $\mathcal{G}_{12|34}$ acts as
\beqn
\sigma P  = M_{\sigma(1)}\otimes M_{\sigma(2)}\otimes M_{\sigma(3)} \otimes M_{\sigma(4)}\cdot \sigma \widetilde{P}\nonumber.
\eqn
Because permuting the transition matrices on the pendant edges will not change which quartet the tensor corresponds to, we need only consider $\sigma \widetilde{P}$, and we need only check the lemma for the elements $(1324)$ and $(13)(24)$, as these form a generating set for $\mathcal{G}_{12|34}$.
Referring to (\ref{eq:trimmed}) and (\ref{eq:perminformal}) we find that
\beqn
(1324)\cdot\widetilde{p}_{i_1i_2i_3i_4} &= \delta_{i_4i_3}\delta_{i_1i_2} m^{(0)}_{i_1i_4}\pi_{i_4},\nonumber
\eqn
and
\beqn
(13)(24)\cdot\widetilde{p}_{i_1i_2i_3i_4} &= \delta_{i_3i_4}\delta_{i_1i_2} m^{(0)}_{i_1i_3}\pi_{i_3}.\nonumber
\eqn
Thus $(1324)\widetilde{P}=(13)(24)\widetilde{P}$, and we see that this tensor belongs to $V^{\mathcal{T}_1}$ (although it corresponds to a quartet rooted at the parent vertex of leaves 3 and 4).

The lemma follows from a similar consideration for phylogenetic tensors belonging to $V^{\mathcal{T}_2}$ and $V^{\mathcal{T}_3}$.
\end{proof}

\begin{figure}[tbp]
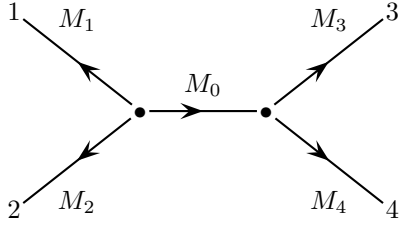

  \centering  
  
  $\psmatrix[colsep=1.5cm,rowsep=.9cm]
  1 & & & 3 \\
    & \bullet &\bullet  \\
  2 & & & 4  
  \psset{ArrowInside=->,nodesep=1pt}
  \psset{arrowscale=2}
  \ncline{2,2}{1,1}^{M_{1}}
  \ncline{2,2}{3,1}_{M_{2}}
  \ncline{2,3}{1,4}^{M_{3}}
  \ncline{2,3}{3,4}_{M_{4}}
  \ncline{2,2}{2,3}^{M_{0}}
  \endpsmatrix
  $
    \caption{Quartet tensor}
    \label{fig:quartettensor}
\end{figure}

We note that there is an obvious analogous structure for the action of $\mathcal{G}_{13|24}$ and $\mathcal{G}_{14|23}$.

Lemma~\ref{lemma1} further illuminates our decision to study isotropy subgroups rather than automorphism groups, and we believe that this reflects the underlying biology of the situation as well.
For instance, it is clear that a phylogenetic method for quartets that returns the quartet tree $12|34$ for a given data set should continue to return $12|34$ even as the input sequences are permuted using elements of $\mathcal{G}_{12|34}$, whereas it is not possible to define an action of $\text{Aut}(\mathcal{T}_1)$ on the input sequences.

\section{Finding tree-informative invariants}
The space of homogeneous degree $d$ polynomials $\mathcal{P}_d(V^{\otimes m})$ carries a representation of $\mathfrak{S}_m$ defined by
\beqn
\sigma^{-1}\circ f(\psi):=f(\sigma\psi),\nonumber
\eqn
with $\psi\in V^{\otimes m}$.
As an example, taking $m\!=\!4$, $d\!=2\!$, we can write 
\beqn
f(\psi)&=\sum_{i_1,\ldots , i_4,j_1,\ldots ,j_4}f_{i_1i_2i_3i_4j_1j_2j_3j_4}\psi_{i_1i_2i_3i_4}\psi_{j_1j_2j_3j_4},\nonumber\\
(123)^{-1}\circ f(\psi)&=\sum_{i_1,\ldots , i_4,j_1,\ldots ,j_4}f_{i_1i_2i_3i_4j_1j_2j_3j_4}\psi_{i_3i_1i_2i_4}\psi_{j_3j_1j_2j_4},
\eqn
so it is apparent that
\beqn
\left[(123)^{-1}\circ f\right]_{i_1i_2i_3i_4j_1j_2j_3j_4}=f_{i_2i_3i_1i_4j_2j_3j_1j_4}.\nonumber
\eqn

From \citet{sumner2008} we know that there exist degree $d\!=\! 5$ Markov invariants for quartet tensors:
\beqn
F:=\left\{f\in \mathcal{P}_5(V^{\otimes 4})\,|\,g^{-1}\circ f=\det(g)f \right\},\nonumber
\eqn 
where $g=M_1\otimes M_2\otimes M_3\otimes M_4$ and each $M_i$ is an element of the Markov semigroup.
Additionally, by considering the inner multiplication of Schur functions it was shown that $\dim(F)\!=\!4$.
Our purpose in the present work is to show how to find linear combinations of these invariants that are tree informative for a given quartet.

\begin{lem}
\label{lemma2}
At a given degree $d$, the subset $W\subset \mathcal{P}_d(V^{\otimes m})$ of phylogenetic invariants for a tree $\mathcal{T}$ is an invariant subspace under the action of $\mathcal{G}_\mathcal{T}$.
\end{lem}

\begin{proof}
Taking $z\in W$, $P\in V^{\mathcal{T}}$, and $\sigma\in \mathcal{G}_\mathcal{T}$ we have 
\beqn
\sigma^{-1}\circ z (P)=z(\sigma P)=0,\nonumber
\eqn
because $\sigma P\in V^{\mathcal{T}}$ by definition.
\end{proof}

For example, it is clear by inspection that the quartet invariants given at the end of \citet{evans1993} form an invariant subspace of $\mathcal{G}_{12|34}$, as required. 
At the end of this section we will examine the invariants given in that work more closely.

In the context of this article we are interested in finding the subspace of Markov invariants that are simultaneously phylogenetic invariants for $\mathcal{T}_1$, i.e. $f\in F$ such that $f(P)=0$ for all $P \in V^{\mathcal{T}_1}$.
As any invariant subspace must occur as a direct sum of irreducible modules, our immediate task is to identify the irreducible representations of $\mathcal{G}_{12|34}$.
For convenience, in Table~\ref{tab:multtab} we present the multiplication table of $\mathcal{G}_{12|34}$.

\begin{table}[t]
\centering
\begin{tabular}{|c|ccccccc|}
\hline
 & (12) & (34) & (12)(34) & (13)(24) & (14)(23) & (1324) & (1423) \\
\hline
(12) & e & (12)(34) & (34) & (1324) & (1423) & (13)(24) & (14)(23) \\
(34) & (12)(34) & e & (12) & (1423) & (1324) & (14)(23) & (13)(24) \\
(12)(34) & (34) & (12) & e & (14)(23) & (13)(24) & (1423) & (1324) \\
(13)(24) & (1423) & (1324) & (14)(23) & e & (12)(34) & (34) & (12) \\
(14)(23) & (1324) & (1423) & (13)(24) & (12)(34) & e & (12) & (34) \\
(1324) & (14)(23) & (13)(24) & (1423) & (12) & (34) & (12)(34) & e \\
(1423) & (13)(24) & (14)(23) & (1324) & (34) & (12) & e & (12)(34) \\
\hline
\end{tabular}
\caption{The group multiplication table of $\mathcal{G}_{12|34}$.}
\label{tab:multtab}
\end{table}

Recall that the irreducible representations of a finite group can be put in one-to-one correspondence with its conjugacy classes, and the sum of the dimension of each irreducible representation squared is equal to the order of the group (see \citet{sagan2001} for example). 
Referring to Table~\ref{tab:multtab}, we go ahead and explicitly compute by hand the conjugacy classes of $\mathcal{G}_{12|34}$.
We find that there are five classes: 
\beqn
\left[e\right]&:=\left\{e\right\},\nonumber\\
[(12)]&:=\left\{(12),(34)\right\},\\
[(12)(34)]&:=\left\{(12)(34)\right\},\\
[(13)(24)]&:=\left\{(13)(24),(14)(23)\right\},\\
[(1324)]&:=\left\{(1324),(1423)\right\},
\eqn
and thus conclude that there are five irreducible representations of $\mathcal{G}_{12|34}$.
It is satisfying to note that this result can be confirmed using the combinatorial formula given in \citet{orellana2004}.

Additionally, we can infer that four of these representations are one-dimensional while the other is two-dimensional, as $1^2+1^2+1^2+1^2+2^2=8$ is the only 5 part partition of 8 into a sum of squares.
We denote the four one-dimensional representations as $\texttt{id}$, $\texttt{sgn}$, $d_1$, $d_2$, and the two-dimensional representation as $C$.

It is useful to note that $(12)(34)$ forms its own conjugacy class. 
This should be compared to the case for $\mathfrak{S}_4$ where $(12)(34),(13)(24)$ and $(14)(23)$ form a single conjugacy class and is due to the fact that $(12)(34)$ is a rotation, while $(13)(24)$ and $(14)(23)$ are reflections.
Using the well known orthogonality relations for characters \citep{sagan2001}, the character table of $\mathcal{G}_{12|34}$ is easy to derive and is presented in Table~\ref{tab:chartab}.

The reader is reminded that the conjugacy classes (and hence irreducible representations) of $\mathfrak{S}_m$ are labelled by partitions of $m$ with $\texttt{id}\equiv (1^m)$ and $\texttt{sgn}\equiv (m)$.
For convenience, we reproduce the character table of $\mathfrak{S}_4$ in Table~\ref{tab:chartabS4}.

Recall that a (group) \emph{branching rule} describes the decomposition of the irreducible representations of a group when restricted to a subgroup \citep{weyl1950}.
By staring at the character tables (Table~\ref{tab:chartab} and Table~\ref{tab:chartabS4}) and concentrating on the conjugacy class $[(12)(34)]$ in $\mathfrak{S}_4$ compared to the same class in $\mathcal{G}_{12|34}$, it is straightforward to derive the group branching rules:
\beqn\label{eq:branchingrules}
\texttt{id} &\rightarrow \texttt{id} \\
 \texttt{sgn} &\rightarrow \texttt{sgn} \\
\mathfrak{S}_4 \downarrow \mathcal{G}_{12|34}: \hspace{2em} \{31\} &\rightarrow C + d_2 \\
 \{2^2\} &\rightarrow \texttt{id} + \texttt{sgn} \\
 \{21^2\} &\rightarrow C + d_1.
\eqn 

\begin{table}[t]
\centering
\begin{tabular}{|c|ccccc|}
\hline
 & $\texttt{id}$ & $\texttt{sgn}$ & $d_1$ & $d_2$ & $C$ \\
\hline
$e$ & 1 & 1 & 1 & 1 & 2 \\
$[(12)]$ & 1 & -1\hspace{.35em} & -1\hspace{.35em} & 1 & 0 \\
$[(12)(34)]$ & 1 & 1 & 1 & 1 & -2\hspace{.35em} \\
$[(13)(24)]$ & 1 & 1 & -1\hspace{.35em} & -1\hspace{.35em} & 0 \\
$[(1324)]$ & 1 & -1\hspace{.35em} & 1 & -1\hspace{.35em} & 0\\
\hline
\end{tabular}
\caption{The character table of $\mathcal{G}_{12|34}$.}
\label{tab:chartab}
\end{table}

\begin{table}[t]
\centering
\begin{tabular}{|c|ccccc|}
\hline
 & $\texttt{id}$ & $\texttt{sgn}$ & $(31)$ & $(2^2)$ & $(21^2)$ \\
\hline
$e$ & 1 & 1 & 3 & 2 & 3 \\
$[(12)]$ & 1 & -1\hspace{.35em} & 1 & 0 & -1\hspace{.35em} \\
$[(123)]$ & 1 & 1 & 0 & -1\hspace{.35em} & 0 \\
$[(12)(34)]$ & 1 & 1 & -1\hspace{.35em} & 2 & -1\hspace{.35em} \\
$[(1234)]$ & 1 & -1\hspace{.35em} & -1\hspace{.35em} & 0 & 1 \\
\hline
\end{tabular}
\caption{The character table of $\mathfrak{S}_4$.}
\label{tab:chartabS4}
\end{table}

Given that $F$ is a module for $\mathfrak{S}_4\downarrow \mathcal{G}_{12|34}$, we would like to examine the structure of Markov invariants in each irreducible module thereof.
This will reveal exactly when an invariant is tree-informative.

Recall that the \emph{primitive idempotents} \citep{procesi2007} of the group algebra $\mathbb{C}\left[\mathcal{G}\right]$ are 
\beqn
\Theta_{\chi}:=\frac{1}{|\mathcal{G}|}\sum_{\sigma \in \mathcal{G}}\chi(\sigma)\sigma,\nonumber
\eqn
where $\chi$ is an irreducible character.
These primitive idempotents satisfy the orthogonality conditions $\Theta_{\chi}\cdot \Theta_{\chi'}=\delta_{\chi\chi'}\Theta_{\chi}$,
and, given a $\mathcal{G}$-module $V$, project onto the irreducible subspaces of $V$.

We are, of course, interested in $\mathcal{G}=\mathcal{G}_{12|34}$ and will consider properties of an arbitrary $f\in F$ under the projections $\Theta_{\chi}\circ f$ for each irreducible character of $\mathcal{G}_{12|34}$.
In what follows we use the fact that $\chi(\sigma^{-1})=\overline{\chi(\sigma)}$ for finite groups, thus
\beqn
\Theta_{\chi}\circ f&=\frac{1}{8}\sum_{\sigma \in \mathcal{G}_{12|34}}\chi(\sigma)\sigma\circ f=\frac{1}{8}\sum_{\sigma \in \mathcal{G}_{12|34}}\chi(\sigma^{-1})\sigma^{-1}\circ f=\frac{1}{8}\sum_{\sigma \in \mathcal{G}_{12|34}}\chi(\sigma)\sigma^{-1}\circ f,\nonumber
\eqn
where the second equality holds because the map $\sigma\mapsto \sigma^{-1}$ is simply a permutation of the group elements and the third equality holds because the irreducible characters of $\mathcal{G}_{12|34}$ are real.

For convenience we take the trimmed tensor $\widetilde{P}_1\in V^{\mathcal{T}_1}$ as before with root placed at the parent vertex of leaves 1 and 2.
This tensor has components
\beqn
\widetilde{p}_{i_1i_2i_3i_4}=\delta_{i_1i_2}\delta_{i_3i_4} m^{(0)}_{i_3i_1}\pi_{i_1}.\nonumber
\eqn 
Define the ``reflected'' trimmed tensor $\widetilde{P}_1^r$ as
\beqn
\widetilde{P}_1^r=(13)(24)\widetilde{P}_1,\nonumber
\eqn
so that $\widetilde{P}_1^r$ is obtained by moving the root vertex to the parent of leaves 3 and 4.
The trimmed tensors $\widetilde{P}_2,\widetilde{P}_3$ and their reflected  counterparts $\widetilde{P}^r_2,\widetilde{P}^r_3$ are defined similarly.
In Table~\ref{tab:actions} we explicitly record the action of $\mathcal{G}_{12|34}$ on each of these trimmed tensors.

\begin{table}
\renewcommand{\arraystretch}{1.4}
\centering
\begin{tabular}[h]{|c|ccc|}
\hline
$\sigma\in\mathcal{G}_{12|34}$ & $\sigma\widetilde{P}_1$ & $\sigma\widetilde{P}_2$ & $\sigma\widetilde{P}_3$  \\
\hline
$e$ & $\widetilde{P}_1$ & $\widetilde{P}_2$ & $\widetilde{P}_3$  \\
$(12)$ & $\widetilde{P}_1$ & $\widetilde{P}_3^r$ & $\widetilde{P}_2^r$  \\
$(34)$ & $\widetilde{P}_1$ & $\widetilde{P}_3$ & $\widetilde{P}_2$  \\
$(12)(34)$ & $\widetilde{P}_1$ & $\widetilde{P}_2^r$ & $\widetilde{P}_3^r$  \\
$(13)(24)$ & $\widetilde{P}_1^r$ & $\widetilde{P}_2$ & $\widetilde{P}_3^r$  \\
$(14)(23)$ & $\widetilde{P}_1^r$ & $\widetilde{P}_2^r$ & $\widetilde{P}_3$  \\
$(1324)$ & $\widetilde{P}_1^r$ & $\widetilde{P}_3^r$ & $\widetilde{P}_2$  \\
$(1423)$ & $\widetilde{P}_1^r$ & $\widetilde{P}_3$ & $\widetilde{P}_2^r$  \\
\hline
\end{tabular}
\caption{Action of $\mathcal{G}_{12|34}$ on trimmed tensors.}
\label{tab:actions}
\end{table}

Now using the character table for $\mathcal{G}_{12|34}$, we can infer any tree-informative identities that occur between the values of $\Theta_{\chi}\circ f(\widetilde{P}_i)$ for $i=1,2,3$ and each irreducible character $\chi$.

For the $\texttt{id}$ representation we have
\beqn
\Theta_{\texttt{id}}\circ f(\widetilde{P}_1)&:=\frac{1}{8}\sum_{\sigma \in \mathcal{G}_{12|34}}\chi_{\texttt{id}}(\sigma)f(\sigma \widetilde{P}_1)\\\nonumber
&=\frac{1}{8}\left[f(\widetilde{P}_1)+f(\widetilde{P}_1)+f(\widetilde{P}_1)+f(\widetilde{P}_1)+f(\widetilde{P}_1^r)+f(\widetilde{P}_1^r)+f(\widetilde{P}_1^r)+f(\widetilde{P}_1^r)\right]\\
&=\frac{1}{2}\left[f(\widetilde{P}_1)+f(\widetilde{P}_1^r)\right],
\eqn
\beqn
\Theta_{\texttt{id}}\circ f(\widetilde{P}_2)&:=\frac{1}{8}\sum_{\sigma \in \mathcal{G}_{12|34}}\chi_{\texttt{id}}(\sigma)f(\sigma \widetilde{P}_2)\\\nonumber
&=\frac{1}{8}\left[f(\widetilde{P}_2)+f(\widetilde{P}_3^r)+f(\widetilde{P}_3)+f(\widetilde{P}_2^r)+f(\widetilde{P}_2)+f(\widetilde{P}_2^r)+f(\widetilde{P}_3^r)+f(\widetilde{P}_3)\right]\nonumber\\
&=\frac{1}{4}\left[f(\widetilde{P}_2)+f(\widetilde{P}_3^r)+f(\widetilde{P}_3)+f(\widetilde{P}_2^r)\right],
\eqn
and
\beqn
\Theta_{\texttt{id}}\circ f(\widetilde{P}_3)&:=\frac{1}{8}\sum_{\sigma \in \mathcal{G}_{12|34}}\chi_{\texttt{id}}(\sigma)f(\sigma \widetilde{P}_3)\\\nonumber
&=\frac{1}{8}\left[f(\widetilde{P}_3)+f(\widetilde{P}_2^r)+f(\widetilde{P}_2)+f(\widetilde{P}_3^r)+f(\widetilde{P}_3^r)+f(\widetilde{P}_3)+f(\widetilde{P}_2)+f(\widetilde{P}_2^r)\right]\nonumber\\
&=\frac{1}{4}\left[f(\widetilde{P}_3)+f(\widetilde{P}_2^r)+f(\widetilde{P}_2)+f(\widetilde{P}_3^r)\right].
\eqn
We see that this representation is not tree-informative.

For the $\texttt{sgn}$ representation we have
\beqn
\Theta_{\texttt{sgn}}\circ f(P_1)&:=\frac{1}{8}\sum_{\sigma \in \mathcal{G}_{12|34}}\chi_{\texttt{sgn}}(\sigma)f(\sigma \widetilde{P}_1)\\
&=\frac{1}{8}\left[f(\widetilde{P}_1)-f(\widetilde{P}_1)-f(\widetilde{P}_1)+f(\widetilde{P}_1)+f(\widetilde{P}_1^r)+f(\widetilde{P}_1^r)-f(\widetilde{P}_1^r)-f(\widetilde{P}_1^r)\right]\\
&=0,\nonumber
\eqn
\beqn
\Theta_{\texttt{sgn}}\circ f(\widetilde{P}_2)&:=\frac{1}{8}\sum_{\sigma \in \mathcal{G}_{12|34}}\chi_{\texttt{sgn}}(\sigma)f(\sigma \widetilde{P}_2)\\
&=\frac{1}{8}\left[f(\widetilde{P}_2)-f(\widetilde{P}_3^r)-f(\widetilde{P}_3)+f(\widetilde{P}_2^r)+f(\widetilde{P}_2)+f(\widetilde{P}_2^r)-f(\widetilde{P}_3^r)-f(\widetilde{P}_3)\right]\\
&=\frac{1}{4}\left[f(\widetilde{P}_2)+f(\widetilde{P}_2^r)-f(\widetilde{P}_3)-f(\widetilde{P}_3^r)\right],\nonumber
\eqn
and
\beqn
\Theta_{\texttt{sgn}}\circ f(\widetilde{P}_3)&:=\frac{1}{8}\sum_{\sigma \in \mathcal{G}_{12|34}}\chi_{\texttt{sgn}}(\sigma)f(\sigma \widetilde{P}_3)\\
&=\frac{1}{8}\left[f(\widetilde{P}_3)-f(\widetilde{P}_2^r)-f(\widetilde{P}_2)+f(\widetilde{P}_3^r)+f(\widetilde{P}_3^r)+f(\widetilde{P}_3)-f(\widetilde{P}_2)-f(\widetilde{P}_2^r)\right]\\
&=\frac{1}{4}\left[f(\widetilde{P}_3)+f(\widetilde{P}_3^r)-f(\widetilde{P}_2)-f(\widetilde{P}_2^r)\right].\nonumber
\eqn
Thus in this case we have $\Theta_{\texttt{sgn}}\circ f(\widetilde{P}_1)=0$ and $\Theta_{\texttt{sgn}}\circ f(\widetilde{P}_2)=-\Theta_{\texttt{sgn}}\circ f(\widetilde{P}_3)$, so that this representation is tree-informative.
A major outcome of this article is that these are exactly the relations that were derived in \citet{sumner2008} by explicit computation.

For the $d_1$ representation we have
\beqn
\Theta_{d_1}\circ f(\widetilde{P}_1)&:=\frac{1}{8}\sum_{\sigma \in \mathcal{G}_{12|34}}\chi_{d_1}(\sigma)f(\sigma \widetilde{P}_1),\\
&=\frac{1}{8}\left[f(\widetilde{P}_1)-f(\widetilde{P}_1)-f(\widetilde{P}_1)+f(\widetilde{P}_1)-f(\widetilde{P}_1^r)-f(\widetilde{P}_1^r)+f(\widetilde{P}_1^r)+f(\widetilde{P}_1^r)\right]\\
&=0,\nonumber
\eqn
\beqn
\Theta_{d_1}\circ f(\widetilde{P}_2)&:=\frac{1}{8}\sum_{\sigma \in \mathcal{G}_{12|34}}\chi_{d_1}(\sigma)f(\sigma \widetilde{P}_2),\\
&=\frac{1}{8}\left[f(\widetilde{P}_2)-f(\widetilde{P}_3^r)-f(\widetilde{P}_3)+f(\widetilde{P}_2^r)-f(\widetilde{P}_2)-f(\widetilde{P}_2^r)+f(\widetilde{P}_3^r)+f(\widetilde{P}_3)\right]\nonumber\\
&=0
\eqn
and
\beqn
\Theta_{d_1}\circ f(\widetilde{P}_3)&:=\frac{1}{8}\sum_{\sigma \in \mathcal{G}_{12|34}}\chi_{d_1}(\sigma)f(\sigma \widetilde{P}_3),\\
&=\frac{1}{8}\left[f(\widetilde{P}_3)-f(\widetilde{P}_2^r)-f(\widetilde{P}_2)+f(\widetilde{P}_3^r)-f(\widetilde{P}_3^r)-f(\widetilde{P}_3)+f(\widetilde{P}_2)+f(\widetilde{P}_2^r)\right]\nonumber\\
&=0.
\eqn
We see that this representation vanishes on \emph{every} quartet. 

For the $d_2$ representation we have 
\beqn
\Theta_{d_2}\circ f(\widetilde{P}_1)&:=\frac{1}{8}\sum_{\sigma \in \mathcal{G}_{12|34}}\chi_{d_2}(\sigma)f(\sigma \widetilde{P}_1),\\
&=\frac{1}{8}\left[f(\widetilde{P}_1)+f(\widetilde{P}_1)+f(\widetilde{P}_1)+f(\widetilde{P}_1)-f(\widetilde{P}_1^r)-f(\widetilde{P}_1^r)-f(\widetilde{P}_1^r)-f(\widetilde{P}_1^r)\right]\\
&=\fra{1}{2}\left[f(\widetilde{P}_1)-f(\widetilde{P}_1^r)\right],\nonumber
\eqn
\beqn
\Theta_{d_2}\circ f(\widetilde{P}_2)&:=\frac{1}{8}\sum_{\sigma \in \mathcal{G}_{12|34}}\chi_{d_2}(\sigma)f(\sigma \widetilde{P}_2),\\
&=\frac{1}{8}\left[f(\widetilde{P}_2)+f(\widetilde{P}_3^r)+f(\widetilde{P}_3)+f(\widetilde{P}_2^r)-f(\widetilde{P}_2)-f(\widetilde{P}_2^r)-f(\widetilde{P}_3^r)-f(\widetilde{P}_3)\right]\nonumber\\
&=0
\eqn
and
\beqn
\Theta_{d_2}\circ f(\widetilde{P}_3)&:=\frac{1}{8}\sum_{\sigma \in \mathcal{G}_{12|34}}\chi_{d_2}(\sigma)f(\sigma \widetilde{P}_3),\\
&=\frac{1}{8}\left[f(\widetilde{P}_3)+f(\widetilde{P}_2^r)+f(\widetilde{P}_2)+f(\widetilde{P}_3^r)-f(\widetilde{P}_3^r)-f(\widetilde{P}_3)-f(\widetilde{P}_2)-f(\widetilde{P}_2^r)\right]\nonumber\\
&=0.
\eqn
This representation vanishes identically on the quartets $13|24$ and $14|23$ but not on $12|34$.
 
As the $C$ representation is 2-dimensional we consider a tuple $f:=(f_1,f_2)\mapsto \Theta_{C}\circ f$ with $f_1,f_2\in F$:
\beqn
\Theta_{C}\circ f(\widetilde{P}_1)&:=\frac{1}{8}\sum_{\sigma \in \mathcal{G}_{12|34}}\chi_{C}(\sigma)f(\sigma \widetilde{P}_1) = \frac{1}{8}\left[2f(\widetilde{P}_1)-2f(\widetilde{P}_1)\right] =0,\nonumber\\
\Theta_{C}\circ f(\widetilde{P}_2)&:=\frac{1}{8}\sum_{\sigma \in \mathcal{G}_{12|34}}\chi_{C}(\sigma)f(\sigma \widetilde{P}_2) =\frac{1}{8}\left[2f(\widetilde{P}_2)-2f(\widetilde{P}_2^r)\right] =\frac{1}{4}\left[f(\widetilde{P}_2)-f(\widetilde{P}_2^r)\right],\nonumber
\eqn
and
\beqn
\Theta_{C}\circ f(\widetilde{P}_3)&:=\frac{1}{8}\sum_{\sigma \in \mathcal{G}_{12|34}}\chi_{C}(\sigma)f(\sigma \widetilde{P}_3) =\frac{1}{8}\left[2f(\widetilde{P}_3)-2f(\widetilde{P}_3^r)\right] =\frac{1}{4}\left[f(\widetilde{P}_3)-f(\widetilde{P}_3^r)\right].\nonumber
\eqn
In this case the representation vanishes identically on the quartet $12|34$ but not on the other two quartets, and is hence tree-informative.

It is worth noting that the above relations are generic statements about invariants that belong to particular irreducible modules of $\mathcal{G}_{12|34}$ and it is still possible for there to be additional tree-informative relations.
For example, in the $\texttt{id}$ case it is clear that an invariant could be tree-informative if it so happened that $f(\widetilde{P}_1)+f(\widetilde{P}_1^r)=0$.

It seems that the tree-informative Markov invariants identified in \citet{sumner2008} transform under the $\texttt{sgn}$ representation of $\mathcal{G}_{12|34}$.
Unfortunately, our understanding of the Schur-Weyl duality does not allow us to take the final step and directly write $F$ as a sum of irreducible modules of $\mathfrak{S}_4$. 
This is because the details of the $\mathfrak{S}_4$ symmetry seems to get lost in the derivation of the existence conditions given in \citet{sumner2008}.
However, in the next section will give a procedure that generates invariants in $F$ that have clear transformation properties under $\mathfrak{S}_4$.
As it will be clear which modules these invariants belong to, we need only give a linearly independent set of four invariants to infer the decomposition of $F$ into irreducible modules of $\mathfrak{S}_4$, and whence of $\mathcal{G}_{12|34}$ using the group branching rules (\ref{eq:branchingrules}).

Before we do this however, we return to the invariants of \citet{evans1993} and explicitly show how they occur as irreducible modules of $\mathcal{G}_{12|34}$.
As an illustration of the power of the present approach, we can even do this without delving into the precise meaning of the formal expressions they give for their invariants.

In Section 7 \citet{evans1993} give phylogenetic invariants for the Kimura 3ST model on the quartet tree $12|34$ in three forms 
\beqn
z^{(a)}(\chi,\chi')&:=\mathbb{E}\left[\left\langle Y_1+Y_2,\chi \right\rangle\left\langle Y_3+Y_4,\chi' \right\rangle\right]-\mathbb{E}\left[\left\langle Y_1+Y_2,\chi \right\rangle\right]\mathbb{E}\left[\left\langle Y_3+Y_4,\chi' \right\rangle\right],\nonumber\\ z^{(b)}(\chi,\chi')&:=\mathbb{E}\left\langle\sum_{i=1}^4Y_{i},\chi\right\rangle\mathbb{E}\left\langle\sum_{i=1}^4Y_{i},\chi'\right\rangle\\
&\hspace{8em}-\mathbb{E}\left[\left\langle Y_1+Y_2,\chi \right\rangle\left\langle Y_3+Y_4,\chi' \right\rangle\right]\mathbb{E}\left[\left\langle Y_1+Y_2,\chi' \right\rangle\left\langle Y_3+Y_4,\chi \right\rangle\right],\\
z^{(c)}(\chi,\chi')&:=\mathbb{E}\left[\left\langle Y_1+Y_3,\chi \right\rangle\left\langle Y_2+Y_4,\chi' \right\rangle\right]\mathbb{E}\left[\left\langle Y_1+Y_2,\chi' \right\rangle\left\langle Y_3+Y_4,\chi \right\rangle\right]\\
&\hspace{9em}-\mathbb{E}\left[\left\langle Y_1+Y_4,\chi \right\rangle\left\langle Y_2+Y_3,\chi' \right\rangle\right]\mathbb{E}\left[\left\langle Y_1+Y_4,\chi' \right\rangle\left\langle Y_2+Y_3,\chi \right\rangle\right],
\eqn
where $\mathbb{E}$ denotes expectation, $Y_1,Y_2,Y_3,Y_4$ are the random variables at the leaves of the quartet taking on values in the abelian group $\mathbb{Z}_2\times \mathbb{Z}_2$ and $\chi$,$\chi'$ are \emph{non-trivial} characters of $\mathbb{Z}_2\times \mathbb{Z}_2$.

Now it is clear by inspection that $z^{(b)}$ transforms as the $\texttt{id}$ representation of $\mathcal{G}_{12|34}$ and $z^{(c)}$ transforms as the $\texttt{sgn}$ representation.
For $z^{(a)}$ we observe that $z^{(a)}(\chi,\chi)$ transforms as the $\texttt{id}$ representation, as does the symmetric combination $z^{(a)}(\chi,\chi')+z^{(a)}(\chi',\chi)$.
Finally, by inspecting Table~\ref{tab:chartab} we see that the anti-symmetric combination $z^{(a)}(\chi,\chi')-z^{(a)}(\chi',\chi)$ transforms as the $d_1$ representation.
In this way we have completely characterized these quartet invariants into irreducible modules of $\mathcal{G}_{12|34}$. 
\section{Explicit forms}
In this section we present a trick that freely generates Markov invariants, and we apply the previous theory to identify which $\mathcal{G}_{12|34}$-module these invariants belong to. 
We conclude by identifying $F$ as a sum of irreducible $\mathcal{G}_{12|34}$-modules.

We begin by observing that the (completely antisymmetric) Levi-Citiva tensor
\beqn 
\epsilon_{i_1i_2i_3i_4}:= \texttt{sgn}(i_1i_2i_3i_4)\nonumber
\eqn 
transforms as the $\texttt{sgn}$ representation of the general linear group $GL(\mathbb{C}^4)$.
That is, for any $g\in GL(\mathbb{C}^4)$,  
\beqn
\sum_{1\leq j_1,j_2j_3,j_4\leq 4}g_{i_1j_1}g_{i_2j_2}g_{i_3j_3}g_{i_4j_4}\epsilon_{j_1j_2j_3j_4}=\det(g) \epsilon_{i_1i_2i_3i_4}.\nonumber
\eqn
Now, in a procedure that is consistent with that given in \citet{sumner2008} (we only ignore symmetrization across the rows of associated tableaux), we can freely construct Markov invariants such as
\beqn
f(\psi)=\sum\psi_{\Sigma\Sigma i_3i_4}\psi_{j_1j_2\Sigma\Sigma}\psi_{k_1k_2k_3k_4}\psi_{l_1l_2l_3l_4}\psi_{m_1m_2m_3m_4}\epsilon_{j_1k_1l_1m_1}\epsilon_{j_2k_2l_2m_2}\epsilon_{i_3k_3l_3m_3}\epsilon_{i_4k_4l_4m_4},\nonumber
\eqn
where each subscript ``$\Sigma$'' can be thought as either a sum over states (as with \citet{allman2003}) or the ``0'' component of the basis specified in \citet{sumner2008}, and all remaining indices are summed from 1 to $k$.
One can readily check that if 
\beqn
\psi_{i_1i_2i_3i_4}\rightarrow \psi'_{i_1i_2i_3i_4}=\sum_{1\leq j_1,j_2,j_3,j_4\leq k}m^{(1)}_{i_1j_1}m^{(2)}_{i_2j_2}m^{(3)}_{i_3j_3}m^{(4)}_{i_4j_4}\psi_{j_1j_2j_3j_4},\nonumber
\eqn 
with each $m^{(a)}_{ij}$ a Markov matrix such that $\sum_{i}m^{(a)}_{ij}=1$, that 
\beqn
f(\psi')=\det(M_1)\det(M_2)\det(M_3)\det(M_4)f(\psi),\nonumber
\eqn
as required.
Note that this construction requires that the $\Sigma$'s are spread evenly across the legs of the tensors (one for each part of the tensor product).

It is worth observing that this presentation can be related to that given by \citet{allman2003} by observing that the cofactor matrix can be expressed as 
\beqn
\left[\text{cof}(M)\right]_{ab}=\sum_{1\leq i_1,i_2,j_1,j_2,k_1,k_2\leq k}m_{i_1i_2}m_{j_1j_2}m_{k_1k_2}\epsilon_{i_1j_1k_1a}\epsilon_{i_2j_2k_2b}.\nonumber
\eqn 
However, in that work the phylogenetic invariants constructed were not required to have any particular transformation properties under the action of the Markov semigroup. 
It would also be of interest to determine the transformation properties of the invariants given in \citet{allman2003} under the relevant isotropy subgroup.

In the general case of $k$ states, the Levi-Citiva tensor has $k$ legs, thus the minimum degree we can construct an invariant as above is $d\!=\!k$.
However, by anti-symmetry this only works for even $m$, and we can construct a single $d\!=\!k$ Markov invariant for each even $m$.
This is consistent with \citet{sumner2008} where it was observed that there exist Markov invariants of degree $d\!=\!k$ for even $m$ only.
For $m\!=\!2$ the corresponding Markov invariant forms the foundation of the Log-Det distance estimator, and $m\!=\!4$ the Markov invariant is referred to as the ``quangle''.

Taking the quartet case $m\!=\!4$ and $d\!=\!k+1$, we must insert a total of four $\Sigma$'s into the expression for the Markov invariant (one for each leg of the tensor product).
If we represent the five factors in the expression as boxes $I$, $J$, $K$, $L$ and $M$, we are asking how many ways are there to put four objects $\{1,2,3,4\}$ into 5 identical boxes.
Clearly, for each set partition of $\{1,2,3,4\}$ this can be done in the various ways given in Table~\ref{tab:setpart}.
For example, we have 
\beqn
f^{(12,34)}(\psi)=\sum\psi_{\Sigma\Sigma i_3i_4}\psi_{j_1j_2\Sigma\Sigma}\psi_{k_1k_2k_3k_4}&\psi_{l_1l_2l_3l_4}\psi_{m_1m_2m_3m_4}\\
&\cdot\epsilon_{j_1k_1l_1m_1}\epsilon_{j_2k_2l_2m_2}\epsilon_{i_3k_3l_3m_3}\epsilon_{i_4k_4l_4m_4},\nonumber
\eqn
and
\beqn
f^{(12,3,4)}(\psi)=\sum \psi_{\Sigma\Sigma i_3i_4}\psi_{j_1j_2\Sigma j_4}\psi_{k_1k_2k_3\Sigma}&\psi_{l_1l_2l_3l_4}\psi_{m_1m_2m_3m_4}\\
&\cdot\epsilon_{j_1k_1l_1m_1}\epsilon_{j_2k_2l_2m_2}\epsilon_{i_3k_3l_3m_3}\epsilon_{i_4j_4l_4m_4}.\nonumber
\eqn
Now given that the rows in each set partition can be interchanged freely, it is easy to check that under $\mathfrak{S}_4$ these invariants transform amongst each other following the permutations, e.g. $\sigma\cdot(ijk,l)=(\sigma(i)\sigma(j)\sigma(k),\sigma(l))$.
In fact one can explicitly check that 
\beqn
(124)\circ f^{(12,34)}=f^{(24,13)}=f^{(13,24)}.\nonumber
\eqn
Thus for each set partition, the corresponding invariants form an invariant subspace of $\mathfrak{S}_4$.
As is depicted in Table~\ref{tab:setpart}, we label these invariant subspaces by enclosing the partition shape within square brackets $\left[\cdot\right]$.
That is,
\beqn
\left[2^2\right]:=\langle f^{(12,34)},f^{(13,24)},f^{(14,23)}\rangle,\nonumber
\eqn
where $\langle \cdot,\ldots,\cdot \rangle$ denotes linear span.

It is too much to hope that for each set partition that the corresponding invariant subspace will be irreducible, but using the primitive idempotents of $\mathfrak{S}_4$ it is a straightforward pencil and paper computation to show that for the $\left[4\right]$ module we have
\beqn
\Theta_{\texttt{id}}\circ f^{(1234)}&=f^{(1234)},\nonumber\\
\Theta_{(31)}\circ f^{(1234)}&=\Theta_{(2^2)}\circ f^{(1234)}=\Theta_{(21^2)}\circ f^{(1234)}=\Theta_{\texttt{sgn}}\circ f^{(1234)}=0.
\eqn
For the $\left[31\right]$ module we note the $\mathfrak{S}_4$ symmetry so we need only consider the canonical example
\beqn
\Theta_{\texttt{id}}\circ f^{(123,4)}&=\fra{1}{4}\left(f^{(123,4)}+f^{(124,3)}+f^{(134,2)}+f^{(234,1)}\right),\nonumber\\
\Theta_{(31)}\circ f^{(123,4)}&=\fra{1}{24}\left(3f^{(123,4)}-f^{(124,3)}-f^{(134,2)}-f^{(234,1)}\right),\\
\Theta_{(2^2)}\circ f^{(123,4)}&=\Theta_{(21^2)}\circ f^{(123,4)}=\Theta_{\texttt{sgn}}\circ f^{(123,4)}=0,
\eqn
with obvious similar relations for $f^{(124,3)},f^{(134,2)}$ and $f^{(234,1)}$.
For the $\left[2^2\right]$ module we can again exploit the $\mathfrak{S}_4$ symmetry and consider
\beqn
\Theta_{\texttt{id}}\circ f^{(12,34)}&=\fra{1}{3}\left(f^{(12,34)}+f^{(13,24)}+f^{(14,23)} \right),\nonumber\\
\Theta_{(2^2)}\circ f^{(12,34)}&=\fra{1}{6}\left(2f^{(12,34)}-f^{(13,24)}-f^{(14,23)} \right),\\
\Theta_{(31)}\circ f^{(12,34)}&=\Theta_{(21^2)}\circ f^{(12,34)}=\Theta_{\texttt{sgn}}\circ f^{(12,34)}=0,
\eqn
with obvious similar relations for $f^{(13,24)}$ and $f^{(14,23)}$.
Similarly, for the $\left[21^2\right]$ module we have
\beqn
\Theta_{\texttt{id}}\circ f^{(12,3,4)}&=\fra{1}{6}\left(f^{(12,3,4)}+f^{(13,2,4)}+f^{(14,2,3)}+f^{(23,1,4)}+f^{(24,1,3)}+f^{(34,1,2)} \right),\nonumber\\
\Theta_{(31)}\circ f^{(12,3,4)}&=\fra{1}{6}\left(f^{(12,3,4)}-f^{(34,1,2)}\right),\\
\Theta_{(2^2)}\circ f^{(12,3,4)}&=\fra{1}{12}\left(2(f^{(12,3,4)}+f^{(34,1,2)})-(f^{(13,2,4)}+f^{(24,1,3)}+f^{(14,2,3)}+f^{(23,1,4)}) \right),\\
\Theta_{(21^2)}\circ f^{(12,3,4)}&=\Theta_{\texttt{sgn}}\circ f^{(12,3,4)}=0.
\eqn
Finally, for the $\left[1^4\right]$ module:
\beqn
\Theta_{\texttt{id}}\circ f^{(1,2,3,4)}&=f^{(1,2,3,4)},\nonumber\\
\Theta_{(31)}\circ f^{(1,2,3,4)}&=\Theta_{(2^2)}f^{(1,2,3,4)}=\Theta_{(21^2)}\circ f^{(1,2,3,4)}=\Theta_{\texttt{sgn}}\circ f^{(1,2,3,4)}=0.\\
\eqn

Thus as irreducible modules of $\mathfrak{S}_4$, we have
\beqn
\left[4\right] &\cong \texttt{id},\nonumber\\
[31] &\cong \texttt{id}\oplus (31),\\
[2^2] &\cong \texttt{id}\oplus (2^2),\\
[21^2] &\cong \texttt{id}\oplus (2^2)\oplus (31),\\
[1^4] &\cong \texttt{id}.
\eqn
It is also worth noting that the dimensions of these modules add up the the number of invariants given in Table~\ref{tab:setpart}. 

\begin{table}
\centering
\begin{tabular}[h]{|l|llllll|}
\hline
I & 1234 & & & & & \\
\hline
I & 123 & 124 & 134 & 234 & &  \\
J &	4 & 3 & 2 & & &  \\
\hline
I & 12 & 13 & 14 & & & \\
J &	34 & 24 & 23 & & & \\
\hline
I & 12 & 13 & 14 & 23 & 24 & 34 \\
J &	3 & 2 & 2 & 1 & 1 & 1 \\
K & 4 & 4 & 3 & 4 & 3 & 2 \\
\hline
I & 1 & & & & & \\
J & 2	& & & & & \\
K & 3 & & & & & \\
L & 4 & & & & & \\
\hline
\end{tabular}
\caption{Classes of invariants: $\left[4\right]$, $\left[31\right]$, $\left[2^2\right]$, $\left[21^2\right]$ and $\left[1^4\right]$.}
\label{tab:setpart}
\end{table}

However, we know that $F$ is only 4 dimensional, so we have far too many invariants. 
To help rectify this, we note that
\beqn
f^{(1234)}(\psi)=\sum\psi_{\Sigma\Sigma\Sigma\Sigma}\psi_{j_1j_2j_3j_4}\psi_{k_1k_2k_3k_4}\psi_{l_1l_2l_3l_4}\psi_{m_1m_2m_3m_4}\epsilon_{j_1k_1l_1m_1}\ldots \epsilon_{j_4k_4l_4m_4},\nonumber
\eqn
which can be factorised into a degree $d\!=\!4$ invariant multiplied by the ``trivial'' invariant $\Phi(\psi):=\psi_{\Sigma\Sigma\Sigma\Sigma}$. 
Thus, $\left[4\right]\in \Phi\cdot \mathcal{P}_{4}(V^{\otimes 4})^{\times^4 GL(V)}$, so we can conclude that
\beqn
F=\left[4\right]\oplus \bar{F},\nonumber
\eqn
with $\dim(\bar{F})=3$.

At this point we throw our hands in the air and resort to explicit computation with \texttt{R} \citep{Rproject} (code available upon request) to show that
\beqn
\left[4\right]&\cong \left[31\right]_{\texttt{id}},\nonumber\\
\left[21^2\right]_{\texttt{id}}&\cong \left[1^4\right],\\
\left[2^2\right]_{\texttt{id}}&\in \left\langle\left[4\right],\left[1^4\right]\right\rangle,\\
\left[2^2\right]_{(2^2)}&\cong \left[21^2\right]_{(2^2)},\\
\left[31\right]_{(31)}&\equiv 0,\\
\left[21^2\right]_{(31)}&\equiv 0,
\eqn
where $\left[\cdot\right]_{(\cdot)}$ denotes the $(\cdot)$ $\mathfrak{S}_{4}$-module contained in $\left[\cdot\right]$.
From this we can conclude that
\beqn
F=\left[4\right]\oplus \left[1^4\right]\oplus \left[2^2\right]_{(2^2)}.\nonumber
\eqn
So that, as a decomposition into irreducible representations of $\mathfrak{S}_4$, we have
\beqn
F=2\cdot \text{\texttt{id}}\oplus (2^2).\nonumber
\eqn

Referring to the branching rule $\mathfrak{S}_4\downarrow \mathcal{G}_{12|34}$, as a decomposition into irreducible modules of $\mathcal{G}_{12|34}$ we see that
\beqn
F=3\cdot \texttt{id}\oplus \texttt{sgn}.\nonumber 
\eqn
Thus we have achieved our main aim of expressing $F$ as a direct sum of irreducible modules of $\mathfrak{S}_4$ and $\mathcal{G}_{12|34}$.

By decomposing $F$ into a direct sum of irreducible modules of $\mathcal{G}_{12|34}$ we have shown that there is a single copy of the $\texttt{sgn}$ representation and hence a single tree-informative Markov invariant for the quartet $\mathcal{T}_1\!:=12|34$.
Using the primitive idempotent of the $(2^2)$ representation of $\mathfrak{S}_4$ we have
\beqn
\Theta_{(2^2)}\circ f^{(13,24)}&=\fra{1}{6}\left(2f^{(13,24)}-f^{(14,23)}-f^{(12,34)} \right).\nonumber
\eqn
Now projecting further with the $\texttt{sgn}$ representation of $\mathcal{G}_{12|34}$ we get
\beqn
\Theta_{\texttt{sgn}}\circ\fra{1}{6}\left(2f^{(13,24)}-f^{(13,23)}-f^{(12,34)} \right)=\fra{1}{2}\left(f^{(13,24)}-f^{(14,23)}\right),\nonumber
\eqn
and we define
\beqn
Q_1:=\fra{1}{2}\left(f^{(13,24)}-f^{(14,23)}\right).\nonumber
\eqn
We can use the action $(14)\cdot \mathcal{T}_1\mapsto \mathcal{T}_2$ to transform this invariant to produce a tree-informative invariant for $\mathcal{T}_2$:
\beqn
Q_2:=\fra{1}{2}\left(f^{(14,23)}-f^{(12,34)}\right),\nonumber
\eqn
and similarly to produce a tree informative invariant for $\mathcal{T}_3$:
\beqn
Q_3:=\fra{1}{2}\left(f^{(12,34)}-f^{(13,24)}\right).\nonumber
\eqn
These are none other than the Markov invariants referred to as the ``squangles'' (\textbf{s}tochastic quangles) in \citet{sumner2008}.

Similar considerations reveal that the three Markov invariants that transform as the $\texttt{id}$ representation of $\mathcal{G}_{12|34}$ are $f^{(1234)}$, $f^{(1,2,3,4)}$, and $f^{(12,34)}$.
We summarize all of this in the following theorem.

\begin{thm}
The set of Markov invariants for quartet trees 
\beqn
F:=\left\{f\in \mathcal{P}_5(V^{\otimes 4})\,|\,g^{-1}\circ f=\det(g)f \right\},\nonumber
\eqn
where $g=M_1\otimes M_2\otimes M_3\otimes M_4$ and each $M_i$ is an element of the Markov semigroup, can be decomposed into irreducible modules of $\mathfrak{S}_4$ as
\beqn
F&=2\cdot \texttt{id}\oplus \left(2^2\right)\\
&=\left\langle f^{(1234)}\right\rangle\oplus \left\langle f^{(1,2,3,4)}\right\rangle\oplus \left\langle 2f^{(12,34)}-f^{(13,24)}-f^{(14,23)},2f^{(13,24)}-f^{(12,34)}-f^{(14,23)}\right\rangle,\nonumber
\eqn
and irreducible modules of $\mathcal{G}_{12|34}$ as
\beqn
F&=3\cdot \texttt{id}\oplus \texttt{sgn}\\
&=\left\langle f^{(1234)}\right\rangle\oplus \left\langle f^{(1,2,3,4)}\right\rangle\oplus \left\langle f^{(12,34)}\right\rangle\oplus \left\langle f^{(13,24)}-f^{(14,23)} \right\rangle.\nonumber
\eqn
\end{thm}

\subsection*{}
As a final loose end, we note that a crucial aspect to the performance of the Markov invariants in the simulation study given in \citet{sumner2008} was the observation that
\beqn
Q_{1}(P_2)\geq 0, \nonumber
\eqn
with similar relations for the other invariants.
Now we have explicit forms for the invariants we can easily derive the relevant relations.
Consider, consistent with $\mathcal{T}_2$, the ``trimmed'' phylogenetic tensor $P$ with components $p_{i_1i_2i_3i_4}=\delta_{i_1i_3}\delta_{i_2i_4}\psi_{i_1i_2}$ where $\psi_{i_1i_2}=\pi_{i_1}m^{(0)}_{i_1i_2}$.
Now
\beqn
f^{(13,24)}(P)&=p_{\Sigma i_2\Sigma i_4}p_{j_1\Sigma j_3\Sigma}p_{k_1k_2k_3k_4}p_{l_1l_2l_3l_4}p_{m_1m_2m_3m_4}\epsilon_{j_1k_1l_1m_1}\epsilon_{i_2k_2l_2m_2}\epsilon_{j_3k_3l_3m_3}\epsilon_{i_4k_4l_4m_4}\\\nonumber
&=\psi_{\Sigma i_2}\delta_{i_2i_4}\psi_{j_1\Sigma}\delta_{j_1j_3}\psi_{k_1k_2}\delta_{k_1k_3}\delta_{k_2k_4}\psi_{l_1l_2}\delta_{l_1l_3}\delta_{l_2l_4}\psi_{m_1m_2}\delta_{m_1m_3}\delta_{m_2m_4}\\
&\hspace{5em}\cdot\epsilon_{j_1k_1l_1m_1}\epsilon_{i_2k_2l_2m_2}\epsilon_{j_3k_3l_3m_3}\epsilon_{i_4k_4l_4},\\
&=\psi_{\Sigma i_2}\psi_{j_1\Sigma}\psi_{k_1k_2}\psi_{l_1l_2}\psi_{m_1m_2}|\epsilon_{j_1k_1l_1m_1}||\epsilon_{i_2k_2l_2m_2}|,
\eqn
and similarly
\beqn
f^{(14,23)}(P)=\psi_{j_1i_2}^2\psi_{k_1k_2}\psi_{l_1l_2}\psi_{m_1m_2}|\epsilon_{j_1k_1l_1m_1}||\epsilon_{i_2k_2l_2m_2}|.\nonumber
\eqn
It is clear that $\psi_{j_1i_2}^2\leq \psi_{\Sigma i_2}\psi_{j_1\Sigma}$ for all $j_1,i_2$, and we have the required result.

With our newly computed forms of the squangles expressed using the Levi-Citiva tensor, we repeated the simulation study given in \citet{sumner2008} and yielded identical results.
This gives a strong experimental confirmation of the theory underlying this work, as the previous forms of the squangles were computed using the Young tableaux procedure given in \citet{sumner2008}.
Also we note that the tree-informative squangles are actually linearly dependent:
\beqn
Q_1+Q_2+Q_3=0.\nonumber
\eqn
This refines the results given in \citet{sumner2008} where this dependence was not observed.
This was missed because of the obscure nature of the basis used in the construction of the Young tableaux. 
Hopefully this article has helped to illuminate some of these issues significantly.

\section{Discussion}
In this article we have applied the representation theory of the isotropy subgroup of leaf permutations on a quartet to give a systematic procedure for finding tree-informative invariants.
In the quartet case we applied this to Markov invariants and reproduced from theoretical considerations relations that were previously derived computationally.

For general unrooted binary trees the corresponding isotropy groups arise as combinations of direct and ``wreath'' products of $\mathfrak{S}_2$ and $\mathfrak{S}_3$.
For example, in the quartet case $\mathcal{G}_{12|34}\cong \mathfrak{S}_2\wr \mathfrak{S}_2$, and for the (balanced) binary tree with 6 leaves and 3 cherries we have $\mathcal{G}_{12|34|56}\cong \mathfrak{S}_2\wr \mathfrak{S}_3$.
It would be fruitful to continue to study the representation theory of wreath product groups with an eye applications to phylogenetic problems. 
In particular, it is worth noting here that the isotropy subgroup for ``caterpillar'' (completely unbalanced) trees is isomorphic to the quartet case.
Thus the theory we have developed in this article will apply directly in that case with complication in detail only, as there are more invariants and more trees to check against for linear relations.
Additionally, for the case of completely balanced \emph{rooted} trees, the irreducible representations have been enumerated in \citet{orellana2004}.

Using leaf permutations we have been able to explicitly incorporate the underlying tree structure into the analysis of tensor-based approaches to phylogenetic problems.
This is surely a step forward, but there remains a gap between the work presented in this article and that presented in \citet{sumner2008}.
That is, one would like to derive the decomposition of the module of Markov invariants into irreducible modules of the tree isotropy groups directly without the need for any explicit computation.
This was not quite achieved in this article and presents itself as an open problem. 

More generally, the opportunity exists to derive a general duality between representations of the Markov semigroup and those of tree isotropy groups. 
This would be in analogy to the Schur-Weyl duality between representations of the general linear and the symmetric group.

\subsection*{Acknowledgement}

We thank an anonymous reviewer for helpful comments.

\subsection*{Role of funding source}

This research was conducted with support from the Australian Research Council Discovery Project grant DP0877447.


\begin{thebibliography}{20}
\expandafter\ifx\csname natexlab\endcsname\relax\def\natexlab#1{#1}\fi
\expandafter\ifx\csname url\endcsname\relax
  \def\url#1{\texttt{#1}}\fi
\expandafter\ifx\csname urlprefix\endcsname\relax\def\urlprefix{URL }\fi
\providecommand{\selectlanguage}[1]{\relax}

\bibitem[{Allman \& Rhodes(2003)}]{allman2003}
\textsc{Allman, E.~S. \& Rhodes, J.~A.} (2003).
\newblock Phylogenetic invariants of the general {Markov} model of sequence
  mutation.
\newblock \emph{Math. Biosci.} \textbf{186}, 113--144.

\bibitem[{Burnham \& Anderson(2002)}]{burnham2002}
\textsc{Burnham, K. P., \& Anderson, D.} (2002).
\newblock \emph{Model Selection and Multi-Model Inference}.
\newblock Springer-Verlag.

\bibitem[{Cavender \& Felsenstein(1987)}]{cavender1987}
\textsc{Cavender, J.~A. \& Felsenstein, J.} (1987).
\newblock Invariants of phylogenies in a simple case with discrete states.
\newblock \emph{J. Class.} \textbf{4}, 57--71.

\bibitem[{Chang(1996)}]{chang1996}
\textsc{Chang, J.~T.} (1996).
\newblock Full reconstruction of {Markov} models on evolutionary trees:
  identifiability and consistency.
\newblock \emph{Math. Biosci.} \textbf{137(1)}, 51--73.

\bibitem[{Eriksson(2008)}]{eriksson2008}
\textsc{Eriksson, N.} (2008).
\newblock Using invariants for phylogenetic tree construction.
\newblock In: \emph{Emerging Applications of Algebraic Geometry}
  (\textsc{Putinar, M. \& Sullivant, S.}, eds.). Springer.

\bibitem[{Evans \& Speed(1993)}]{evans1993}
\textsc{Evans, S.~N. \& Speed, T.~P.} (1993).
\newblock Invariants of some probability models used in phylogenetic inference.
\newblock \emph{Ann. Stat.} \textbf{21(1)}, 355--377.

\bibitem[{Felsenstein(2004)}]{felsenstein2004}
\textsc{Felsenstein, J.} (2004).
\newblock \emph{Inferring Phylogenies}.
\newblock Sinauer Associates.

\bibitem[{Gascuel(2005)}]{gascuel2005}
\textsc{Gascuel, O.} (ed.)  (2005).
\newblock \emph{Mathematics of Evolution and Phylogenetics}.
\newblock Oxford University Press.

\bibitem[{Gawron \emph{et~al.}(1999)Gawron, Nekrashevic \&
  Sushchanskii}]{gawron1999}
\textsc{Gawron, P., Nekrashevic, V.~V. \& Sushchanskii, V.~I.} (1999).
\newblock Conjugacy classes of the automorphism group of a tree.
\newblock \emph{Mathematical Notes} \textbf{65}, 787--790.

\bibitem[{Godsil \& Royle(2001)}]{godsil2001}
\textsc{Godsil, C. \& Royle, G.} (2001).
\newblock \emph{Algebraic Graph Theory}.
\newblock Graduate Text in Mathematics. Springer-Verlag.

\bibitem[{Hendy \& Penny(1989)}]{hendy1989}
\textsc{Hendy, M.~D. \& Penny, D.} (1989).
\newblock A framework for the quantitative study of evolutionary trees.
\newblock \emph{Syst. Zool.} \textbf{38}, 297--309.

\bibitem[{Lake(1987)}]{lake1987}
\textsc{Lake, J.~A.} (1987).
\newblock A rate-independent technique for analysis of nucleic acid sequences:
  evolutionary parsimony.
\newblock \emph{Mol. Biol. Evol.} \textbf{4}, 167--191.

\bibitem[{Matsen \& Steel(2007)}]{matsen2007}
\textsc{Matsen, F.~A. \& Steel, M.} (2007).
\newblock Phylogenetic mixtures on a single tree can mimic a tree of another
  topology.
\newblock \emph{Syst. Biol.} \textbf{56}, 767--775.

\bibitem[{Orellana \emph{et~al.}(2004)Orellana, Orrison \&
  Rockmore}]{orellana2004}
\textsc{Orellana, R.~C., Orrison, M.~E. \& Rockmore, D.~N.} (2004).
\newblock Rooted trees and iterated wreath products of cyclic groups.
\newblock \emph{Adv. Appl. Math.} \textbf{33}, 531--547.

\bibitem[{Procesi(2007)}]{procesi2007}
\textsc{Procesi, C.} (2007).
\newblock \emph{Lie Groups: An Approach through Invariants and
  Representations}.
\newblock Springer.

\bibitem[{{R Development Core Team}(2006)}]{Rproject}
\textsc{{R Development Core Team}} (2006).
\newblock \emph{R: A Language and Environment for Statistical Computing}.
\newblock R Foundation for Statistical Computing, Vienna, Austria.

\bibitem[{Sagan(2001)}]{sagan2001}
\textsc{Sagan, B.~E.} (2001).
\newblock \emph{The Symmetric Group: Representations, Combinatorial Algorithms,
  and Symmetric Functions. Second Edition.}
\newblock Graduate Texts in Mathematics. Springer.

\bibitem[{Semple \& Steel(2003)}]{semple2003}
\textsc{Semple, C. \& Steel, M.} (2003).
\newblock \emph{Phylogenetics}.
\newblock Oxford Press.

\bibitem[{Steel(1994)}]{steel1994}
\textsc{Steel, M.~A.} (1994).
\newblock Recovering a tree from the leaf colourations it generates under a
  {Markov} model.
\newblock \emph{Appl. Math. Lett.} \textbf{7}, 19--24.

\bibitem[{Sumner \emph{et~al.}(2008)Sumner, Charleston, Jermiin \&
  Jarvis}]{sumner2008}
\textsc{Sumner, J.~G., Charleston, M.~A., Jermiin, L.~S. \& Jarvis, P.~D.}
  (2008).
\newblock Markov invariants, plethyms and phylogenetics.
\newblock \emph{J. Theor. Biol.} \textbf{253}, 601--615.

\bibitem[{Weyl(1950)}]{weyl1950}
\textsc{Weyl, H.} (1950).
\newblock \emph{The Theory of Groups and Quantum Mechanics}.
\newblock Dover Publications.

\end{thebibliography}
\end{document}